\documentclass[11pt,letterpaper]{article}

\usepackage[margin=1in]{geometry}
\usepackage[T1]{fontenc}
\usepackage[utf8]{inputenc}
\usepackage{amsmath,amssymb,amsthm,mathtools}
\usepackage{newtxtext,newtxmath}
\usepackage{microtype}
\usepackage[hyphens]{url}
\usepackage{xurl}
\usepackage[authoryear,round]{natbib}
\usepackage{enumitem}
\usepackage{hyperref}

\newcommand{\ArxivAuthorName}{Hiroaki Odahara}
\newcommand{\ArxivAffiliation}{%
Market Design Center, Graduate School of Economics, The University of Tokyo\\
Graduate School of Informatics and Engineering, The University of Electro-Communications}
\newcommand{\ArxivDate}{July 2026}

\hypersetup{
  colorlinks=true,
  linkcolor=black,
  citecolor=black,
  urlcolor=black,
  pdftitle={The Targeted-Loss Exposure Frontier in Auctions},
  pdfauthor={\ArxivAuthorName},
  pdfsubject={Auction theory and mechanism design},
  pdfkeywords={revenue equivalence, targeted-loss exposure, pay-to-bid auctions, shill bidding, linear programming, silent war of attrition}
}

\newtheorem{theorem}{Theorem}[section]
\newtheorem{proposition}[theorem]{Proposition}
\newtheorem{lemma}[theorem]{Lemma}
\newtheorem{corollary}[theorem]{Corollary}

\newcommand{\E}{\mathbb{E}}
\newcommand{\TLE}{\operatorname{TLE}}
\newcommand{\WOA}{\mathrm{WOA}}
\newcommand{\AP}{\mathrm{AP}}
\newcommand{\dd}{\,\mathrm{d}}

\title{The Targeted-Loss Exposure Frontier in Auctions}
\author{\ArxivAuthorName\\[0.35em]
\small \ArxivAffiliation}
\date{\ArxivDate}

\begin{document}
\maketitle

\begin{abstract}
Pay-to-bid auctions charge participants before allocation and therefore make losing-side payments vulnerable to seller intervention. This paper introduces targeted-loss exposure, a stress test that records a designated bidder's payment when one rival's equilibrium report is fixed at an arbitrarily high level while her type distribution is preserved. Under nonnegative payments and an ordering in which a winner pays at least as much as a loser, revenue equivalence yields a sharp two-bidder bound attained by the silent war of attrition. For any number of bidders, the same format attains the upper bound among payment-ordered rank-local rules. The proof formulates payment location as a linear program and constructs a dual probability measure. With more than two bidders, complementary slackness determines the optimal loser-payment schedule almost everywhere. Winner-pay auctions have zero exposure, standard all-pay auctions lie strictly below the frontier, and frontier exposure decreases as the number of bidders rises. A loser-only rule shows that removing payment ordering can make exposure unbounded. A common-shock extension also shows that anticipating intervention changes bids but preserves the ordering of standard formats. The results separate the total interim payment fixed by revenue equivalence from the winning or losing state to which that payment is attached.
\end{abstract}



\medskip
\section{Introduction}

In pay-to-bid (penny) auctions, participants sink a fee with every bid before allocation is resolved \citep{ByersEtAl2010,Augenblick2016,Hinnosaar2016}. In March 2011, the Consumer Affairs Agency issued orders to three Japanese penny-auction operators over misleading representations concerning auction content and transaction terms \citep{CAA2011}. In a separate December 2012 case, police arrested operators on suspicion of fraud, alleging that automated fictitious bids prevented genuine participants from winning while bid fees accumulated \citep{ITmedia2012}. Together these episodes motivate an architectural question that remains current in work on credible auctions and seller manipulation \citep{AkbarpourLi2020,PorterShoham2005,KomoEtAl2025}: how much of a participant's payment sits on the losing side, exposed to an organizer who can influence who loses?

This note isolates that quantity. The targeted-loss experiment fixes one rival's report at an increasingly high equilibrium level---a static shill-bid stress test---while preserving the designated bidder's type distribution; targeted-loss exposure records her payment in states where that report guarantees her loss. Unlike equilibrium loser revenue, it does not weight each type by her equilibrium loss probability. It is a payment diagnostic, not a model of fraud, reputation, or bidder naivety. A common-shock extension lets bidders anticipate forced loss with any positive probability below certainty; bids adjust without reversing standard-format exposure rankings.

With two bidders, nonnegative payments, payment ordering, and revenue equivalence together bound targeted-loss payment pointwise by interim payment. At the top-rival limit, exposure is at most expected value. The silent war of attrition---the sealed-bid second-price all-pay auction in which losers pay their own bids and the winner pays the highest losing bid---attains the bound; without payment ordering, loser-only exposure can be unbounded. The bridge to dynamic penny auctions is architectural rather than temporal. Bid fees accumulate along a participant's own bidding history before allocation, while the silent war of attrition compresses that account into the loser's own bid. A high rival report therefore guarantees loss without erasing the target's committed payment. For any bidder count, the main result identifies maximal exposure among payment-ordered rank-local rules. Writing exposure as the average loser fee identifies the objective, not its feasible set: revenue equivalence and payment ordering jointly constrain how each type's payment budget can be split between winning and losing states. The dual argument establishes attainment and, with more than two bidders, almost-everywhere uniqueness of the loser-payment schedule.

\citet{KrishnaMorgan1997} compare the silent war of attrition and the standard all-pay auction, while \citet{MurtoValimaki2017} study the former with many bidders. The present result instead treats payment location as a design variable and characterizes the silent war of attrition as maximally exposed within the payment-ordered rank-local class. Related criteria include ambiguity-robust transfers \citep{HwangKohBaik2026}, seller-revenue variance \citep{BojkoEtAl2024}, and shill-proofness against seller profit incentives \citep{KomoEtAl2025}. Here the criterion is targeted-loss exposure; formats that charge only winners provide the zero-exposure benchmark.

\section{Targeted-loss exposure and the two-bidder bound}

Let $n\ge2$ risk-neutral bidders have quasilinear utilities and values $V_1,\ldots,V_n$ drawn independently from a distribution $F$ on $[0,1]$ with continuous positive density $f$; write $V\sim F$ for a generic draw. All payment functions below are conditional expected payments given the report profile. Deterministic transfers are a special case; for randomized transfers, nonnegativity and payment ordering are imposed on these conditional expectations. Let $M$ be a symmetric $n$-bidder auction in which the highest report wins, and fix a symmetric strictly increasing equilibrium reporting strategy $\beta$ on $[0,1)$. Designate bidder $i$ and one rival $j$. Let $r\in(0,1)$ index the rival type whose equilibrium report is imposed. The targeted-loss experiment fixes bidder $j$'s report at $\beta(r)$, while bidder $i$ and all remaining rivals report according to $\beta$. Because $\beta$ is strictly increasing, the fixed report guarantees bidder $i$'s loss whenever $V_i<r$. Let $P_i^{M,r}$ denote bidder $i$'s random payment in this experiment. Define
\begin{equation}
\begin{aligned}
L_n^{\mathrm{TL}}(M;r)
&:=\E\!\left[P_i^{M,r}\mathbf 1\{V_i<r\}\right],\\
\TLE_n(M)
&:=\limsup_{r\uparrow1}L_n^{\mathrm{TL}}(M;r).
\end{aligned}
\label{eq:tledef}
\end{equation}
Here $L_n^{\mathrm{TL}}(M;r)$ is the targeted-loss payment at finite threshold $r$. Its upper limit as the imposed rival type approaches the top of the support is targeted-loss exposure, $\TLE_n(M)$. The expectation is over the unfixed bidders' values and any mechanism randomization. The indicator $V_i<r$ selects states in which the fixed report guarantees bidder $i$'s loss, so the expectation is interventional rather than conditional on equilibrium loss. The limsup permits divergent top bids with finite expected payments.

We first set $n=2$, with bidder 1 designated and bidder 2 fixed. Writing $b_i(x_1,x_2)$ for bidder $i$'s conditional expected payment at report profile $(x_1,x_2)$, Equation~\eqref{eq:tledef} becomes
\begin{equation}
L_2^{\mathrm{TL}}(M;r)=\int_0^r b_1\!\left(\beta(v),\beta(r)\right)\dd F(v).
\label{eq:tle2}
\end{equation}

For the two-bidder analysis, call such an auction admissible if payments are nonnegative, the winner pays weakly more than the loser at every report profile, and the lowest type obtains zero expected utility. Efficiency and the envelope formula fix type $r$'s expected payment conditional on her type---the interim payment \citep{Myerson1981,RileySamuelson1981}:
\begin{equation}
p(r)=rF(r)-\int_0^rF(t)\dd t.
\label{eq:p2}
\end{equation}

\begin{theorem}[Pointwise bound]\label{thm:pointwise}
For every admissible auction and every $r\in(0,1)$,
\[
L_2^{\mathrm{TL}}(M;r)\le p(r).
\]
Consequently, $\TLE_2(M)\le p(1)=\E[V]$, and the silent war of attrition attains the top-type bound.
\end{theorem}

\begin{proof}
When $v<r$, bidder 1 loses and bidder 2 wins. Payment ordering gives
$b_1(\beta(v),\beta(r))\le b_2(\beta(v),\beta(r))$.
By symmetry, the right-hand side is the winning-state payment of a type-$r$ bidder facing type $v$. Its integral over $v<r$ is at most that type's total interim payment $p(r)$ because all remaining payments are nonnegative. Continuity of $p$ gives the endpoint bound. In the silent war of attrition,
\[
\beta(r)=\int_0^r\frac{t}{1-F(t)}\dd F(t)
\le-\log(1-F(r)),
\]
and
\[
p(r)-L_2^{\mathrm{TL}}(M;r)=(1-F(r))\beta(r)
\le -(1-F(r))\log(1-F(r))\longrightarrow0.
\]
Hence the top-type bound is attained even though the equilibrium bid diverges at the endpoint.
\end{proof}

Payment ordering has bite. Appendix~\ref{app:ordering} gives an efficient loser-only rule that preserves the revenue-equivalent interim payment but has infinite exposure under uniform values.

\section{The rank-local exposure frontier}

Return to arbitrary $n\ge2$, let $G(v)=F(v)^{n-1}$ be the distribution of the highest rival value, and let
\begin{equation}
p(v)=vG(v)-\int_0^vG(t)\dd t=\int_0^v t\dd G(t)
\label{eq:pn}
\end{equation}
be the revenue-equivalent interim payment \citep{Myerson1981,RileySamuelson1981}. Maintain the assumptions above and consider efficient symmetric formats with a symmetric strictly increasing equilibrium. Let $\mathcal M^{\mathrm{RL}}$ denote the class of such formats whose equilibrium-path payments are rank local. Using equilibrium types as labels, rank locality means that a losing type $v$ pays a nonnegative, nondecreasing $\ell(v)$ depending only on her own type, while a winning type $v$ facing highest rival type $m<v$ pays $W(v,m)$, depending only on $v$ and $m$. Payments are nonnegative, and payment ordering is $W(v,m)\ge\ell(m)$. The lowest type's expected utility is zero. Both schedules are measurable on their domains and have finite expected payments. This class specifies equilibrium-path payments, not off-path details of a full direct mechanism; it contains first-price, Vickrey, standard all-pay, and the silent war of attrition.

Under rank locality, Equation~\eqref{eq:tledef} reduces by monotone convergence to
\begin{equation}
\TLE_n(M)=\int_0^1\ell(v)\dd F(v).
\label{eq:tlen}
\end{equation}
Equilibrium loser revenue would instead weight the integrand by the type's loss probability. Equation~\eqref{eq:tlen} removes that weight while preserving the designated bidder's type distribution.

Type $v$ loses with probability $1-G(v)$, so revenue equivalence implies
\[
p(v)=(1-G(v))\ell(v)+\int_0^vW(v,m)\dd G(m).
\]
Because $W(v,m)\ge\ell(m)$, every $M\in\mathcal M^{\mathrm{RL}}$ induces a feasible schedule in the following relaxation over nonnegative measurable $\ell$:
\begin{equation}
\begin{aligned}
\sup_{\ell\ge0}\quad &\int_0^1\ell(v)\dd F(v)\\
\text{s.t.}\quad &(T\ell)(v):=
\begin{cases}
(1-G(v))\ell(v)+\displaystyle\int_0^v\ell(t)\dd G(t),&0\le v<1,\\[4pt]
\displaystyle\int_0^1\ell(t)\dd G(t),&v=1,
\end{cases}\\[-2pt]
&T\ell\le p.
\end{aligned}
\label{eq:primal}
\end{equation}
The endpoint definition avoids assigning a value to $0\cdot\ell(1)$. The endpoint constraint follows from the interior constraints by monotone convergence. Program~\eqref{eq:primal} drops monotonicity, the explicit winner-payment schedule, and remaining implementability restrictions, so it is an upper-bounding linear-programming relaxation rather than a characterization of the entire feasible set.

Rank locality captures fees tied to a participant's own bidding. Without it, loser payments can vary across rival-report slices that have zero equilibrium probability but are selected by the intervention, changing exposure without changing equilibrium outcomes. Rank locality makes the criterion invariant to such specifications. The bound may fail in broader classes, whose frontier remains open; Theorem~\ref{thm:pointwise}, by contrast, needs no rank-local restriction.

The equilibrium bid in the $n$-bidder silent war of attrition---and its loser-payment schedule---is
\begin{equation}
\beta_n^{\WOA}(v)=\int_0^v\frac{t}{1-G(t)}\dd G(t),\qquad v<1.
\label{eq:woa}
\end{equation}

\begin{lemma}[Primal equality]\label{lem:binds}
The silent-war-of-attrition schedule satisfies $T\beta_n^{\WOA}=p$ on $[0,1]$, and the resulting exposure is finite.
\end{lemma}

\begin{proof}
On the interior, $\dd\beta_n^{\WOA}(v)=v\dd G(v)/(1-G(v))$. The product rule for Stieltjes differentials gives
\[
\dd(T\beta_n^{\WOA})
=\dd\!\left[(1-G)\beta_n^{\WOA}\right]+\beta_n^{\WOA}\dd G
=v\dd G=\dd p.
\]
Both functions vanish at zero, so equality holds below the endpoint. Moreover,
\[
\beta_n^{\WOA}(v)\le-\log(1-G(v)),
\]
which implies $(1-G(v))\beta_n^{\WOA}(v)\to0$ as $v\uparrow1$. Taking limits in the interior equality and applying monotone convergence gives $(T\beta_n^{\WOA})(1)=p(1)$. Finally,
\[
\int_0^1\beta_n^{\WOA}(v)\dd F(v)
\le \int_0^1-\log(1-u^{n-1})\dd u<\infty.
\]
\end{proof}

\begin{theorem}[Exposure frontier]\label{thm:frontier}
For every $n\ge2$,
\[
\sup_{M\in\mathcal M^{\mathrm{RL}}}\TLE_n(M)
=\max_{\ell\ge0:T\ell\le p}\int_0^1\ell(v)\dd F(v)
=\int_0^1\beta_n^{\WOA}(v)\dd F(v).
\]
The standard $n$-bidder silent war of attrition implements the relaxed upper bound. For $n>2$, every optimal measurable schedule in the relaxed program coincides with $\beta_n^{\WOA}$ $\dd G$-almost everywhere. This is uniqueness of the loser-payment schedule, not of off-path winner payments. For $n=2$, the optimum is not unique even within $\mathcal M^{\mathrm{RL}}$; Appendix~\ref{app:n2} constructs a continuum of implementing formats.
\end{theorem}

\begin{proof}
Define the nonincreasing dual tail function
\[
\Lambda(v)=\left(\sum_{k=0}^{n-2}F(v)^k\right)^{-1},
\qquad v\in[0,1].
\]
For $v<1$, it equals $(1-F(v))/(1-G(v))$, and $\Lambda(1)=1/(n-1)$. Let
\[
\mu=-\dd\Lambda+\frac{1}{n-1}\delta_1,
\]
where $\delta_1$ is the unit point mass at one, so that $\mu$ is a probability measure with tail $\mu([v,1])=\Lambda(v)$. The identity $(1-G)\Lambda=1-F$ implies, as Stieltjes measures on the closed interval,
\begin{equation}
\dd F(v)=(1-G(v))\mu(\dd v)+\Lambda(v)\dd G(v).
\label{eq:dualidentity}
\end{equation}
The endpoint atom of $\mu$ is multiplied by $1-G(1)=0$, while $F$ and $G$ have no endpoint atom.

For every feasible schedule, Equation~\eqref{eq:dualidentity}, the tail representation of $\Lambda$, and Tonelli's theorem give
\begin{align*}
\int_0^1\ell(v)\dd F(v)
&=\int_{[0,1)}(1-G(v))\ell(v)\mu(\dd v)
  +\int_0^1\ell(t)\Lambda(t)\dd G(t)\\
&=\int_{[0,1]}(T\ell)(v)\mu(\dd v)\\
&\le\int_{[0,1]}p(v)\mu(\dd v).
\end{align*}
The second equality separates the endpoint atom and uses $\Lambda(t)=\mu([t,1])$; it does not require a value for $\ell(1)$. Lemma~\ref{lem:binds} makes the inequality bind for $\beta_n^{\WOA}$. A loser in the silent war of attrition pays $\beta_n^{\WOA}(v)$, while a winner facing highest rival $m$ pays $\beta_n^{\WOA}(m)$. If type $v$ instead uses the bid prescribed for type $z$, differentiating her expected payoff and using Equation~\eqref{eq:woa} gives $G'(z)(v-z)$. Hence truthful bidding is optimal. The format belongs to $\mathcal M^{\mathrm{RL}}$ and implements the relaxed optimum, proving equality in value.

For $n>2$, the interior part of $\mu$ has strictly positive density $\eta=-\Lambda'$. Equality therefore forces $T\ell=p$ almost everywhere in the interior. For an optimal schedule, put $h=\ell-\beta_n^{\WOA}$ and $H(v)=\int_0^v h(t)\dd G(t)$. Complementary slackness gives, $\dd G$-almost everywhere,
\[
(1-G(v))h(v)+H(v)=0,
\qquad
\dd H=-\frac{H}{1-G}\dd G.
\]
On every interval $[0,a]$ with $a<1$,
\[
\dd\!\left(\frac{H}{1-G}\right)=0.
\]
Since $H(0)=0$, $H$ vanishes on every such interval and hence $h=0$ $\dd G$-almost everywhere. Appendix~\ref{app:n2} constructs a continuum of distinct formats in $\mathcal M^{\mathrm{RL}}$ that attain the optimum when $n=2$.
\end{proof}

\begin{corollary}[Standard formats and bidder count]\label{cor:standard}
Let $\mathrm{FP}$, $\mathrm{Vickrey}$, $\mathrm{AP}$, and $\mathrm{WOA}$ denote first-price, Vickrey, standard all-pay, and silent-war-of-attrition formats. For every nondegenerate $F$,
\begin{align}
\TLE_n(\WOA)
&=(n-1)\int_0^1F^{-1}(u)
\frac{u^{n-2}(1-u)}{1-u^{n-1}}\dd u,
\label{eq:woavalue}\\
\TLE_n(\AP)
&=(n-1)\int_0^1F^{-1}(u)u^{n-2}(1-u)\dd u.
\label{eq:apvalue}
\end{align}
Consequently,
\[
0=\TLE_n(\mathrm{FP})=\TLE_n(\mathrm{Vickrey})
<\TLE_n(\AP)<\TLE_n(\WOA)\le\E[V],
\]
where the final inequality binds only for two bidders. Moreover, $\TLE_n(\WOA)$ and $\beta_n^{\WOA}(v)$ for every interior $v$ decrease strictly with bidder count.
\end{corollary}

\begin{proof}
Fubini's theorem and the substitution $u=F(v)$ give Equation~\eqref{eq:woavalue}. In standard all-pay, the equilibrium bid equals the revenue-equivalent interim payment, and the same substitution gives Equation~\eqref{eq:apvalue}. Winner-pay formats have zero loser payment. The war-of-attrition kernel is the all-pay kernel multiplied by $1/(1-u^{n-1})>1$, proving the strict interior comparison. Also,
\[
\frac{(n-1)u^{n-2}(1-u)}{1-u^{n-1}}
=\frac{(n-1)u^{n-2}}{1+u+\cdots+u^{n-2}}\le1,
\]
with equality throughout only when $n=2$. Finally, write the same kernel as
\[
k_n(u):=\left(\frac{1}{n-1}\sum_{j=0}^{n-2}u^{-j}\right)^{-1}.
\]
Its denominator increases strictly when the bidder count rises. This proves exposure monotonicity. Moreover,
\[
(\beta_n^{\WOA})'(v)=\frac{vf(v)}{1-F(v)}k_n(F(v)),
\]
so the same comparison proves bid monotonicity.
\end{proof}

The two-bidder contest is therefore the most exposed configuration on the frontier. In that case the intervention fixes the sole rival's report, and the dual measure collapses to a top-type point mass. With more bidders, its positive interior density pins down the loser-payment schedule almost everywhere.

\section{Implications}

The frontier identifies payment architectures that most expose participants when commitment to honest execution is weak. Targeted-loss exposure is a payment stress test, not a welfare or damages measure. Within the payment-ordered rank-local class, exposure cannot exceed the designated bidder's expected value. Appendix~\ref{app:ordering} shows that a loser-only rule outside payment ordering can instead have infinite exposure, so the finite frontier is a conservative benchmark for such formats. Attaching payment to winning, or making losing-side fees credibly refundable, lowers exposure by construction.

The common-shock extension in Appendix~\ref{app:shock} also clarifies the behavioral interpretation. Allowing bidders to anticipate forced loss changes equilibrium bids and makes the top war-of-attrition bid finite, but it preserves the standard-format ordering for every positive intervention probability below certainty. The benchmark therefore does not depend on assigning exactly zero probability to intervention.

Revenue equivalence fixes each type's average payment, not whether payment is attached to winning or losing. Payment location is therefore the design variable behind the targeted-loss exposure frontier.

\clearpage
\clearpage
\appendix
\begin{center}
{\Large\bfseries Technical Appendix}
\end{center}
\addcontentsline{toc}{section}{Technical Appendix}
\section{Additional results}

\subsection{Nonuniqueness within the rank-local class}\label{app:n2}

For $n=2$, $G=F$ and the dual measure is $\delta_1$. Hence the objective equals the endpoint constraint:
\[
(T\ell)(1)=\int_0^1\ell(v)\dd F(v).
\]
Any feasible schedule that makes this constraint bind is optimal. To exhibit nonuniqueness, write $\beta=\beta_2^{\WOA}$ and, for $a\in(0,1)$, define
\[
B(a)=\int_0^a\beta(v)\dd F(v),
\qquad
c_a=\frac{B(a)}{1-F(a)},
\]
and
\[
\ell_a(v)=
\begin{cases}
0,&0\le v<a,\\
\beta(v)+c_a,&a\le v<1.
\end{cases}
\]
For $v<a$, $(T\ell_a)(v)=0$. For $a\le v<1$, Lemma~\ref{lem:binds} gives
\[
(T\ell_a)(v)=p(v)-B(a)+c_a(1-F(a))=p(v).
\]
The same equality holds at the endpoint. Thus every $\ell_a$ is feasible and optimal but differs from the silent-war-of-attrition schedule on a set of positive measure.

These schedules are implementable, not merely solutions of the relaxation. For $0\le m<v<1$, define the winner-payment schedule
\[
W_a(v,m)=
\begin{cases}
p(v)/F(v),&0<v<a,\\
\ell_a(m),&a\le v<1.
\end{cases}
\]
The unused value at $v=0$ may be set to zero. At the top report, set $W_a(1,m)=\ell_a(m)$; as in the silent war of attrition, no finite loser fee at one is required because a top report cannot strictly lose. Since $\beta$ is increasing and $c_a\ge0$, $\ell_a$ is nonnegative and nondecreasing, and both schedules are measurable. When $v<a$, every lower rival type also lies below $a$, so $\ell_a(m)=0$ and payment ordering holds. When $v\ge a$, it holds with equality. Moreover, $p(v)/F(v)\le v$, so all displayed winner payments are finite. If a bidder reports $z$, her interim expected payment is
\[
R_a(z)=(1-F(z))\ell_a(z)+\int_0^z W_a(z,m)\dd F(m)=p(z).
\]
For $z<a$, this follows from $\ell_a(z)=0$ and the constant winner payment $p(z)/F(z)$; for $z\ge a$, it follows from $T\ell_a=p$. Hence a type $x$ reporting $z$ obtains expected utility
\[
U_a(x,z)=F(z)x-p(z),
\qquad
\frac{\partial U_a(x,z)}{\partial z}=f(z)(x-z).
\]
Truthful reporting is therefore globally optimal, and the lowest type earns zero. The direct-report format with loser payment $\ell_a$ and winner payment $W_a$ belongs to $\mathcal M^{\mathrm{RL}}$. Since different values of $a$ produce different loser-payment schedules, the frontier is nonunique within the original rank-local class when $n=2$.

\subsection{Why payment ordering matters}\label{app:ordering}

Remove payment ordering and consider the loser-only rule---the $n$-bidder version of the sad-loser auction discussed by \citet{RileySamuelson1981}---with
\[
\ell^{\mathrm{LO}}(v)=\frac{p(v)}{1-G(v)},
\qquad
W^{\mathrm{LO}}(v,m)=0.
\]
If type $v$ mimics type $z$, her expected payoff is
\[
u(v,z)=G(z)v-(1-G(z))\ell^{\mathrm{LO}}(z)=G(z)v-p(z).
\]
Since the derivative with respect to $z$ has the sign of $v-z$, truthful reporting is globally optimal. The allocation is efficient, the lowest type earns zero, and the interim expected payment remains $p(v)$. Under uniform values,
\[
\ell^{\mathrm{LO}}(v)
=\frac{(n-1)v^n}{n(1-v^{n-1})}
\sim\frac{1}{n(1-v)},
\]
so its targeted-loss exposure is infinite. Rank locality alone therefore does not yield a finite frontier.

\subsection{Anticipated discretion under a common shock}\label{app:shock}

The baseline evaluates equilibrium bidding without intervention and then applies the targeted-loss experiment. This subsection instead lets two genuine bidders anticipate a seller-controlled shock. After bids are committed, an unbeatable shill receives the object with probability $\varepsilon\in[0,1)$, so both genuine bidders lose; otherwise the announced rule is executed. The shock is independent of values and bids. Bidders first play the corresponding symmetric equilibrium, after which the same targeted-loss experiment is applied. Write $\TLE_{2,\varepsilon}(M)$ for the resulting exposure.

\begin{proposition}[Anticipated discretion]\label{prop:shock}
For first-price and Vickrey auctions,
\[
\TLE_{2,\varepsilon}(\mathrm{FP})
=\TLE_{2,\varepsilon}(\mathrm{Vickrey})=0.
\]
For standard all-pay,
\[
\beta_\varepsilon^{\AP}(v)=(1-\varepsilon)\int_0^v t\dd F(t),
\qquad
\TLE_{2,\varepsilon}(\AP)
=(1-\varepsilon)\frac12\E[\min\{V_1,V_2\}].
\]
For the silent war of attrition,
\[
(\beta_\varepsilon^{\WOA})'(v)
=\frac{(1-\varepsilon)vf(v)}
{\varepsilon+(1-\varepsilon)(1-F(v))},
\qquad \beta_\varepsilon^{\WOA}(0)=0,
\]
and
\[
\TLE_{2,\varepsilon}(\WOA)
=\int_0^1\frac{(1-\varepsilon)t f(t)(1-F(t))}
{\varepsilon+(1-\varepsilon)(1-F(t))}\dd t.
\]
For every $\varepsilon\in(0,1)$,
\[
0=\TLE_{2,\varepsilon}(\mathrm{FP})
=\TLE_{2,\varepsilon}(\mathrm{Vickrey})
<\TLE_{2,\varepsilon}(\AP)
<\TLE_{2,\varepsilon}(\WOA)<\E[V].
\]
Both positive exposures decrease strictly with $\varepsilon$, and the silent-war-of-attrition exposure converges to $\E[V]$ as $\varepsilon$ decreases to zero.
\end{proposition}

\begin{proof}
Let $\Pi(z\mid v)$ be the expected payoff of type $v$ from using the bid prescribed for type $z$. In first-price,
\[
\Pi(z\mid v)=(1-\varepsilon)F(z)(v-\beta(z)),
\]
so the common factor leaves the equilibrium bid unchanged. In Vickrey, a forced loser pays zero and truthful bidding remains weakly dominant conditional on execution. In all-pay,
\[
\Pi(z\mid v)=(1-\varepsilon)F(z)v-\beta(z),
\]
whose first-order condition gives the stated bid; Fubini gives its exposure.

In the silent war of attrition,
\begin{align*}
\Pi(z\mid v)
&=(1-\varepsilon)F(z)v
-[\varepsilon+(1-\varepsilon)(1-F(z))]\beta(z)\\
&\quad -(1-\varepsilon)\int_0^z\beta(t)f(t)\dd t.
\end{align*}
Differentiation gives the displayed differential equation, and Fubini gives the exposure formula. Apart from the common factor $tf(t)$, the all-pay and war-of-attrition exposure kernels are
\[
(1-\varepsilon)(1-F(t))
\quad\text{and}\quad
\frac{(1-\varepsilon)(1-F(t))}
{\varepsilon+(1-\varepsilon)(1-F(t))}.
\]
The second is strictly larger than the first, strictly below one, and strictly decreasing in $\varepsilon$. Monotone convergence yields the zero-shock limit.
\end{proof}

\section*{Disclosures}

\paragraph{Funding.}
This work was supported by Japan Society for the Promotion of Science KAKENHI (JP21H04979) and Japan Science and Technology Agency ERATO (JPMJER2301). The funders had no role in the study design, analysis, writing, or decision to submit the article.

\paragraph{Competing interests.}
The author has no relevant financial or non-financial interests to disclose.

\paragraph{Data availability.}
No empirical data were used or generated in this study.

\paragraph{Acknowledgments.}
This paper supersedes the targeted-loss portion of UTMD Working Paper No.~116; a companion paper studies participation risk and entry screening. I thank Hitoshi Matsushima, Atsushi Iwasaki, Michihiro Kandori, Daisuke Oyama, Kenjiro Asami, Kazuyuki Higashi, Satoshi Kasamatsu, Daiki Kishishita, Satoshi Nakada, Munenori Nakasato, Shunya Noda, and Kyohei Okumura for comments.

\paragraph{Use of generative AI.}
The author used OpenAI ChatGPT and Codex for language editing, organization, LaTeX preparation, bibliographic organization, and verification-code assistance. The author reviewed all outputs and takes full responsibility for the paper.

\bibliographystyle{plainnat}
\bibliography{references}

\end{document}